\documentclass[11pt,letterpaper]{article}

\usepackage{hyperref}
\usepackage[numbers]{natbib}
\usepackage{xspace}
\usepackage{xcolor}
\usepackage{graphicx}
\usepackage{amsmath,amssymb,amsthm}
\usepackage[margin=1in]{geometry}

\theoremstyle{plain}
\newtheorem{proposition}{Proposition}
\newtheorem{theorem}{Theorem}
\newtheorem{lemma}{Lemma}
\newtheorem{corollary}{Corollary}
\theoremstyle{definition}

\newcommand{\secref}[1]{Section~\ref{#1}}
\newcommand{\thmref}[1]{Theorem~\ref{#1}}
\newcommand{\lemref}[1]{Lemma~\ref{#1}}
\newcommand{\proref}[1]{Proposition~\ref{#1}}

\newcommand{\appref}[1]{Appendix~\ref{#1}}

\newcommand{\ie}{i.e.,\xspace}
\newcommand{\eg}{e.g.,\xspace}

\newcommand{\reals}{\mathbb{R}}
\newcommand{\E}{\mathbb{E}}
\newcommand{\midd}{:}
\newcommand{\rgeqk}{\reals_{\smash\geq}^k}

\newcommand{\rbetak}{\reals_{\smash\beta}^k}
\newcommand{\gfp}{$\text{GFP}$\xspace}
\newcommand{\gsp}{$\text{GSP}$\xspace}
\newcommand{\gfpa}{$\text{GFP}_{\smash\alpha}$\xspace}
\newcommand{\gspa}{$\text{GSP}_{\smash\alpha}$\xspace}
\newcommand{\vcg}{$\text{VCG}$\xspace}
\newcommand{\vcga}{$\text{VCG}_{\smash\alpha}$\xspace}

\begin{document}

\title{Expressiveness and Robustness of First-Price Position Auctions\thanks{We have benefitted greatly from discussions with Jason Hartline, Robert Kleinberg, and \'Eva Tardos.}}

\author{%
	Paul D\"{u}tting\thanks{%
	Department of Computer Science, Cornell University, 136 Hoy Road, Ithaca, NY 14850, USA. Email: \texttt{paul.duetting@cornell.edu}. Research supported by an SNF Postdoctoral Fellowship.}
	\and 
	Felix Fischer\thanks{%
	Statistical Laboratory, University of Cambridge, Wilberforce Road, Cambridge CB3 0WB, UK. Email: \texttt{fischerf@statslab.cam.ac.uk}.}
	\and David C.~Parkes\thanks{%
	School of Engineering and Applied Sciences, Harvard University, 33 Oxford Street, Cambridge, MA 02138, USA. Email: \texttt{parkes@eecs.harvard.edu}.}
}

\date{}

\maketitle


\begin{abstract}
	Since economic mechanisms are often applied to very different instances of the same problem, it is desirable to identify mechanisms that work well in a wide range of circumstances.  We pursue this goal for a position auction setting and specifically seek mechanisms that guarantee good outcomes under both complete and incomplete information.  A variant of the generalized first-price mechanism with multi-dimensional bids turns out to be the only standard mechanism able to achieve this goal, even when types are one-dimensional. The fact that expressiveness beyond the type space is both necessary and sufficient for this kind of robustness provides an interesting counterpoint to previous work on position auctions that has highlighted the benefits of simplicity. From a technical perspective our results are interesting because they establish equilibrium existence for a multi-dimensional bid space, where standard techniques break down.  The structure of the equilibrium bids moreover provides an intuitive explanation for why first-price payments may be able to support equilibria in a wider range of circumstances than second-price payments.
\end{abstract}


\section{Introduction}

	Economic mechanisms are often applied to very different instances of the same problem. It is therefore desirable to find mechanisms that work well in a wide range of circumstances, and specifically do not require any knowledge of agents' preferences on the part of the designer. This goal has been formulated many times and forms the core of the Wilson doctrine~\citep{Wils87a} and of the agenda of robust mechanism design~\citep{BeMo05a}.
	
	We pursue this goal for a position auction setting with one-dimensional types, and specifically seek mechanisms that guarantee good outcomes under both complete and incomplete information. Each of~$k$ positions is to be assigned to exactly one of~$n$ agents, and the value of agent~$i$ for position~$j$ can be written as $\beta_j\cdot v_i$, where $v\in\reals^n$ and $\beta\in\rgeqk=\{x\in\reals^k\midd x_j>0,\text{$x_{j}\geq x_{j'}$ if $j<j'$}\}$. In the complete information case~$v$ is common knowledge among the agents. In the incomplete information case the components of~$v$ are independent and identically distributed according to a continuous distribution with bounded support that is common knowledge among the agents, and~$v_i$ is known to agent~$i$. In both cases,~$\beta$ is common knowledge among the agents.  A prime example of this setting can be found in the context of sponsored search, where agents correspond to advertisers, positions correspond to slots in which advertisements can be displayed,~$\beta_j$ denotes the fraction of cases where an advertisement in position~$j$ leads to a conversion, and~$v_i$ denotes the value agent~$i$ has for a conversion.
	
	The goal of the designer is twofold: to provide the best possible service to the agents by allocating positions in a way that maximizes social welfare, \ie the sum of valuations for the allocated positions; and to maximize revenue subject to this constraint. While the agents agree with the former goal, their interests are diametrically opposed to that of the designer with regard to the latter.  From the point of view of the designer, a good mechanism must therefore guarantee existence of an efficient equilibrium and achieve high revenue in \emph{every} efficient equilibrium.  An appropriate revenue benchmark for efficient equilibria is provided by the truthful equilibrium of the well-known Vickrey-Clarkes-Groves (VCG) mechanism~\citep{KrPe00a}.
	
	We arrive at the following question: 
	\begin{quote}
	Does there exist a single mechanism that possesses an efficient equilibrium under both complete and incomplete information, and recovers the truthful VCG revenue in every efficient equilibrium?	
	\end{quote}
	
	To answer this question we consider the three mechanisms commonly used in for position auctions, the VCG mechanism, the generalized first-price (GFP) mechanism, and the generalized second-price (GSP) mechanism.\footnote{Google and Microsoft use the GSP mechanism, Facebook the VCG mechanism. The GFP mechanism was used by Overture, the first company to provide a successful sponsored search service.}  The variants of these three mechanisms we consider all assign the positions from top to bottom to an agent with maximum bid among those not assigned one of the higher positions.\footnote{We use this greedy allocation rule rather than one that selects an efficient allocation relative to the bids because it simplifies the analysis and thus enables our main positive result. All negative results also hold for the efficient allocation rule, and the two allocation rules obviously agree in any efficient equilibrium.} The VCG mechanism charges each agent the externality it imposes on the other agents, the GFP mechanism charges the agent's bid on the position it is allocated, and the GSP mechanism charges the next-highest bid on that position. For each mechanism we moreover distinguish two variants: an expressive variant in which agent $i$ submits a bid $b_{ij}$ for each position~$j$, and a simplified variant in which agent $i$ specifies a single bid $b_i$ and this bid is multiplied by $\alpha\in\rgeqk$ to obtain bids for the different positions. The vector $\alpha$ is part of the mechanism, so it is common knowledge among the agents and the designer and may or may not be identical to $\beta$.

\paragraph{Our contribution}

	It turns out that most candidate mechanisms are disqualified by prior work. The expressive VCG and GSP mechanisms have an efficient complete information equilibrium with revenue zero for all possible valuations of the agents~\citep{Milg10a}.\footnote{This result requires that the agents can bid arbitrary non-negative numbers. It can be circumvented by forcing the agents to submit non-increasing bids. But then there are still efficient equilibria with revenue arbitrarily smaller than the truthful VCG revenue \citep{DFP11a}.} For the simplified variants of these mechanisms the situation is somewhat better, and this has in fact been used as an argument in favor of simplification~\citep{Milg10a}. However, revenue in an efficient complete-information equilibrium may still be arbitrarily small compared to the truthful VCG revenue~\citep{DFP11a}, and the simplified GSP mechanism may not have an equilibrium at all when information is incomplete~\citep{GoSw09a}. The simplified GFP mechanism, on the other hand, has a unique equilibrium under incomplete information~\citep{GoSw09a} but may not have an equilibrium under complete information~\citep{EOS07a}. 

	This only leaves the expressive GFP mechanism, and we show that it indeed possesses the desired robustness property: an efficient equilibrium under both complete and incomplete information, and the truthful VCG revenue in every efficient equilibrium. While good outcomes under either complete or incomplete information can be obtained with a simplified mechanism, expressiveness thus turns out to be both necessary and sufficient for robustness. This provides an interesting counterpoint to previous work on position auctions that has highlighted the benefits of simplicity~\citep{Milg10a,DFP11a}.  An additional advantage of the expressive GFP mechanism is that it is independent of $\beta$. It can therefore be used in settings where the designer is uncertain about the exact value of $\beta$, and our results extend to such settings.

	Our analysis of the complete information case is similar to the classic analysis of \citet{BeWh86} that links equilibria of first-price auctions to the core, and to more recent approaches that also make this connection~\citep{DaMi08a,HJW13a}.  The common feature is the use of what \citet{Milg04a} has called target-profit strategies. Specifically, we show that having each agent $i$ bid its value $\beta_j \cdot v_i$ for position $j$ minus its truthful VCG utility $u_i$, or zero if this is negative, yields an efficient equilibrium. Notable differences concern our use of a greedy rather than efficient allocation rule, and the fact that we show revenue in every efficient equilibrium to be at least the truthful VCG revenue. Unlike prior work we also explicitly handle ties in choosing an allocation.

	As types are one-dimensional, our incomplete information analysis can use \citeauthor{Myer81}'s classic characterization result~\citep{Myer81} to identify equilibrium candidates. The standard technique to verify that a particular candidate is an equilibrium involves integrating the derivative of an agent's utility, as a function of both valuation and bid, along a path between two bids. This technique breaks down in our setting, where the bid space has higher dimension than the valuation space and the utility function may not be defined everywhere on the path. We overcome these difficulties by performing an induction from the last position to the first, and showing that the conjectured equilibrium bid on position $j$ is optimal for agent $i$ given that the other agents bid according to the conjectured equilibrium, and agent $i$ bids according to the conjectured equilibrium on positions $j+1$ to $k$.  We believe that similar techniques can be used to show equilibrium existence in more general settings, including settings with multi-dimensional types.

	Each step of the induction considers only one dimension of the bid space and can use the standard technique, but identifying the equilibrium bids and deriving the utility function is a non-trivial task.  \citeauthor{Myer81}'s theorem only provides a necessary condition for bids that lead to an efficient equilibrium, namely that payments in expectation equal the truthful VCG payments. Since the bid space is multi-dimensional, many different bids satisfy this condition. In the eventual equilibrium, the bid of agent $i$ on position $j$ equals its expected truthful VCG payment conditioned on being allocated position $j$.   These bids again have a natural interpretation in terms of target-profit strategies and also provide an intuitive explanation for why first-price payment rules may be able to support equilibria in a wider range than second-price payment rules: the expected truthful VCG payment of an agent subject to allocation of a given position depends on the agent's valuation and on the distribution from which the valuations of the other agents are drawn, which is exactly the information available to the agent.

\paragraph{Related work}

	The design of more expressive mechanisms for specific applications is an important topic of contemporary mechanism design~\citep[\eg][]{AMPP09a,GhSa10a,CRHP11a,DHW11a,DLN12a,GML12a,GML13a}. In addition, it has been argued more abstractly that the expressiveness of a mechanism is positively correlated with the quality of the outcomes it is able to support. \citet{BSS08a} showed that for combinatorial auctions, the maximum social welfare over all outcomes of a mechanism strictly increases with expressiveness, for a particular measure of expressiveness based on notions from computational learning theory. Implicit in this result is the intuition that more expressiveness is generally desirable, as it allows a mechanism to achieve a more efficient outcome in more instances of the problem.

	The classic analysis of position auctions is due to \citet{Vari07a} and \citet{EOS07a}. Follow-up work has emphasized the benefits of simplicity in this context. \citet{Milg10a} and \citet{DFP11a} considered a complete information setting and showed that simplification can eliminate zero-revenue equilibria without introducing new, and potentially undesirable, equilibria. The authors also pointed out certain advantages of the GSP mechanism over the VCG mechanism in this regard. \citet{GoSw09a} and \citet{ChHa13a} showed that under complete information the GSP mechanism may fail to have an efficient equilibrium, whereas the GFP mechanism always possesses a unique equilibrium, which is efficient and yields the truthful VCG revenue. \citet{PaTa10a}, \citet{LuPa11}, \citet{CKKK11}, and \citet{SyTa13a} showed that the GSP mechanism has a small constant price of anarchy under both complete and incomplete information.\footnote{The price of anarchy compares the minimum social welfare in any equilibrium to the maximum social welfare of any outcome. That greedy algorithms can achieve a small price of anarchy, potentially smaller than that of an efficiently computable outcome, was previously highlighted by \citet{Gair09} in the context of covering games.}   \citet{LPT12a} established lower bounds on the revenue of the GSP mechanism: for complete information it can be arbitrarily small compared to the truthful VCG revenue, for incomplete information it always is a constant fraction of the latter. 

	Our work also has connections to the literature on non-parametric Bayes-Nash implementation, robust full implementation, and prior-free approximation. Non-parametric Bayes-Nash implementation shows that an uninformed designer can implement essentially the same outcomes in equilibrium as an informed designer~\citep[\eg][]{MaSj02}. Robust full implementation seeks to obtain mechanisms that implement a desired outcome in every equilibrium and for any prior the agents may have~\citep[\eg][]{BeMo09a}. Prior-free approximation seeks to approximate a desired outcome for any prior~\citep[\eg][]{DRY10a}.

	To the best of our knowledge, the study of mechanisms for position auctions that admit efficient equilibria and yield high revenue in every efficient equilibrium under both complete and incomplete information, and the use of additional expressiveness to achieve this goal, are both novel.


\section{Preliminaries}

	We study a setting with a set $\{1,\dots,k\}$ of \emph{positions} ordered by quality and a set $N=\{1,\dots,n\}$ of \emph{agents} with \emph{unit demand} and \emph{one-dimensional valuations} for the positions. More formally, write $\rgeqk=\{x\in\reals^k\midd x_j>0,\text{$x_{j}\geq x_{j'}$ if $j<j'$}\}$ for the set of $k$-dimensional vectors whose entries are positive and non-increasing. For $\beta\in\rgeqk$, let $\rbetak=\{x\in\reals^k\midd x=\beta v, v\in\reals_{+}\}$ be the one-dimensional subspace of $\reals^n$ spanned by $\beta\in\rgeqk$. Agent $i$'s valuation can then be represented by a vector $\beta v_i\in\rbetak$ in this subspace, such that $\beta_j v_i\geq 0$ is the agent's value for position $j$. Our goal is to assign the positions to agents in order to maximize total value. We refer to an assignment of agents to positions that achieve this as \emph{efficient}.  Because the base $\beta$ of the subspace $\rbetak$ is the same for all agents, this can be achieved by allocating positions in decreasing order of $v_i$. We assume that $\beta$ is common knowledge among the agents. 

	We compare two kinds of auctions.  An {\em expressive auction} solicits a vector $b_i\in\rgeqk$ of bids from each agent $i\in N$, where $b_{i,j}$ is interpreted as agent $i$'s bid on slot $j$.  A {\em simplified auction}\footnote{We refer to these mechanism as simplified as they can be viewed as resulting from the expressive mechanism by restricting the $k$-dimensional bid space to a $1$-dimensional subspace.} is parameterized by vector $\alpha\in\rgeqk$ and solicits a single-dimensional bid $b_i\in \reals_{+}$ from each agent $i\in N$.  The single-dimensional bid is extended to a $k$-dimensional bid by multiplying it with $\alpha$. Agent $i$'s bid on slot $j$ is thus $\alpha_j v_i.$ 

	More specifically, we consider simplified and expressive variants of the generalized first-price (GFP), generalized second-price (GSP), and the Vickrey-Clarke-Groves (VCG) auctions. We denote these auctions by \gfpa, \gspa, and \vcga and by \gfp, \gsp, and \vcg. We assume that all three mechanism assign the items greedily. That is, starting from the first position and proceeding to the last position, they assign the current position to the agent with the highest bid that has not yet received a position. We focus on greedy winner determination algorithms because it simplifies the equilibrium analysis, and also because it is consistent with the current practice in sponsored search. The payment rules are defined identically for the simplified and expressive variants of each auction. In the generalized first-price auctions, the payment of agent $i$ assigned position $j$ is equal to the bid value associated with position $j$. In the generalized second-price auctions, the payment of agent $i$ for position $j$ is equal to the next-lower bid for that position. In the \vcg auctions, agent $i$ assigned position $j$ is charged an amount equal to the total loss in value of all other agents, according to their bids, caused by assigning position $j$ to agent $i$.

	We make the usual assumption of quasi-linear utilities, such that the utility $u_i(b,v_i)$ of agent $i$ with value $v_i$, in a given auction and for a given bid profile $b$, is equal to its valuation for the position it is assigned minus its payment for that position. To be able to reason about the strategic behavior of agents we need to specify what the agents know about each others' valuations.  In the {\em complete information} setting values $v_i$ are common knowledge among the agents. A vector of bids $(b_1,\dots,b_n)$ is a \emph{Nash equilibrium} of a given mechanism if no agent has an incentive to change its bid assuming that the other agents don't change their bids, \ie if for every $i \in N$ and every~$b'_i$, 
\[
	u_i((b_1,\dots,b_i,\dots,b_n),v_i) \ge u_i((b_1,\dots,b'_i,\dots,b_n),v_n).
\]
In {\em incomplete information} environments, values $v_i$ are drawn independently from a distribution $F$ supported on $[0,\bar{v}]$ for some finite $\bar{v}\in\reals_+$. Distribution $F$ is assumed to be common knowledge among the agents. In this setting, a vector $(b_1,\dots,b_n)$ of bidding functions is a \emph{Bayes-Nash equilibrium} of a given auction if no agent has an incentive to change its bidding function assuming that the other agents don't change their bidding functions and values of the other agents are drawn from $F$, \ie if for every $i\in N$, every $v_i\in[0,\bar{v}]$, and every bidding function $b'_i$,
\begin{multline*}
	\E_{v_j\sim F,j\neq i} \Bigl[ u_i\bigl((b_{1}(v_{1}),\dots,b_{i-1}(v_{i-1}),b_{i}(v_{i}),b_{i+1}(v_{i+1}),\dots,b_{n}(v_{n})),v_i\bigr) \Bigr] \geq \\
	\E_{v_j\sim F,j\neq i} \Bigl[ u_i\bigl((b_{1}(v_{1}),\dots,b_{i-1}(v_{i-1}),b'_{i}(v_{i}),b_{i+1}(v_{i+1}),\dots,b_{n}(v_{n})),v_i\bigr) \Bigr] .
\end{multline*}

	Because our environment is one-dimensional we can appeal to \citeauthor{Myer81}'s characterization of the expected payments in a Bayes-Nash equilibrium.
\begin{theorem}[\citet{Myer81}]\label{thm:myerson}
	Consider a position auction, and assume that agents use bidding functions such that agent $i$ with valuation $v_i$ is assigned position $s$ with probability $P_{s}^i(v_i)$. Then the bidding functions are a Bayes-Nash equilibrium of the auction only if, for every agent $i$,
\begin{enumerate}
\item the expected allocation $\sum_{s=1}^{k} P_s^i(v_i) \beta_s$ is non-decreasing in $v_i$
\item the expected payment is
\[
	p_i(v_i) = \sum_{s=1}^{k} P^i_s(v_i) \beta_s v_i - \int_{0}^{v_i} \sum_{s=1}^{k} P^i_s(z) \beta_s \,dz,
\]
where $p_i(0)=0$.
\end{enumerate}
\end{theorem}

	Since an efficient allocation satisfies monotonicity, we have the following corollary.
\begin{corollary} \label{cor:equiv}
	In an efficient Bayes-Nash equilibrium of any position auction, the expected payment of every agent $i$ is equal, for every value $v_i$, to the expected payment of the agent in the truthful equilibrium of the expressive \vcg auction.
\end{corollary}


\section{Complete Information Analysis}\label{app:complete}
	We begin our analysis by reviewing the properties of the expressive \gfp mechanism in settings with complete information. We show that expressive \gfp always has a Nash equilibrium, that all its Nash equilibria are efficient, and that payments in every Nash equilibrium are at least the truthful VCG payments. The proof is given in \appref{app:complete-information}.

\begin{proposition}\label{pro:complete-information}
Assume that valuations are taken from $\rbetak$. Then, 
\begin{enumerate}
\item the expressive \gfp mechanism has an efficient Nash equilibrium with payments equal to the truthful VCG payments, 
\item every Nash equilibrium of the expressive \gfp mechanism is efficient, and
\item the payments in every Nash equilibrium of the expressive \gfp mechanism are at least the truthful VCG payments.
\end{enumerate}
\end{proposition}



\section{Incomplete Information Analysis}\label{sec:expressive}

	Next we consider environments with incomplete information and show our main result, that the expressive \gfp mechanism always has an efficient equilibrium that yields the truthful VCG revenue.
\begin{theorem} \label{thm:expressive}
	Assume that valuations are drawn independently from a continuous distribution on $\rbetak$ with bounded support. Then the expressive \gfp mechanism has an efficient Bayes-Nash equilibrium with the same payments as the truthful equilibrium of the VCG auction.
\end{theorem}

	We prove this result by constructing a bidding function $b^*:\reals\rightarrow\rgeqk$ and showing by induction that an agent maximizes its utility by bidding according to $b^*$ assuming that all other agents bid according to $b^*$ as well. To this end, we define in \secref{sec:bidding} a function $b_j^*:\reals\rightarrow\reals$ for each position $j$ that maps a valuation~$v$ to the expected truthful VCG payment $b_j^*(v)$ an agent with valuation $v$ would face if it was allocated position~$j$. The equilibrium bidding function $b^*$ will then be given by $b^*(v)=(b^*_1(v),\dots,b^*_k(v))$. We will say that an agent with valuation~$v$ bids truthfully on position $j$ (according to $b^*_j$) if he bids $b_j^*(v)$, and that he bids truthfully if he bids truthfully on all positions. The property we show by induction is that independently of the bids on positions $1,\dots,j-1$ and assuming truthful bids on positions $j+1,\dots,k$, it is optimal to bid truthfully on position $j$. For this we apply the usual technique and integrate the derivative of the utility function from the truthful bid on position~$j$ to a conjectured beneficial deviation on position~$j$ to derive a contradiction.

	Denote by $u^*((x_1,\dots,x_k),v)$ the expected utility of an agent with valuation $v$ who bids $b^*_j(x_j)$ on position $j$ while all other agents bid truthfully.  The proof of \thmref{thm:expressive} uses the following lemmata, which we prove in Sections~\ref{sec:lemma1} and \ref{sec:lemma2}.
\begin{lemma} \label{lem:part-a}
	Fix a particular agent. Assume that all other agents bid truthfully and that the agent bids truthfully on positions $j+1,\dots,k$. Then the derivative in the bid on position $j$ of the agent's expected utility vanishes at the truthful bid, \ie
	\[
		\frac{d}{dx_j} u^*((x_1,\dots,x_j,v,\dots,v),v) \Bigr|_{x_j=v} = 0 .
	\]
\end{lemma}

\begin{lemma}\label{lem:part-b}
	Fix a particular agent. Assume that all other agents bid truthfully and that the agent bids truthfully on positions $j+1,\dots,k$. Then, the derivative in the valuation of the derivative in the bid on position $j$ of the agent's expected utility is non-negative, \ie
	\[
		\frac{d}{dv}\frac{d}{dx_j} u^*((x_1,\dots,x_j,v,\dots,v),v) \geq 0 .
	\]
\end{lemma}

\begin{proof}[Proof of \thmref{thm:expressive}]
	Fix a particular agent and assume that all other agents bid truthfully. 
	Suppose that we have established the claim for positions $j+1,\dots,k$, and that we want to establish it for position $j$. The claim trivially holds for $j=k$, so we know from the induction hypothesis that
	\[
		u^*(x_1,\dots,x_j,v,\dots,v) = \max \{x_{j+1},\dots,x_k\midd u^*(x_1,\dots,x_j,x_{j+1},\dots,x_k)\} .
	\]
	To show that the claim holds for position $j$, assume for contradiction that there exists $v'\in\reals$ such that
	\[
		u^*(x_1,\dots,x_{j-1},v',v,\dots,v) > u^*(x_1,\dots,x_{j-1},v,v,\dots,v) .
	\]
	First assume that $v'<v$. Then,
\begin{align*}
	u^*((x_1,\dots,x_{j-1},v,\dots,v),v) - u^*((x_1,\dots,x_{j-1},v',v,\dots,v),v) & = \\
	\int_{v'}^{v}\! \frac{d}{dx_j}\, u^*((x_1,\dots,x_{j-1},x_j,v,\dots,v),v)\, dx_j & \geq \\
	\int_{v'}^{v}\! \frac{d}{dy}\, u^*((x_1,\dots,x_{j-1},y,x_j,\dots,x_j),x_j) \Bigr|_{y=x_j} dx_j  & = 0 ,
\end{align*}
where the inequality and the last equality respectively hold by \lemref{lem:part-b} and \lemref{lem:part-a}. This is a contradiction.
It is important to note here that when all other agents bid according to $b^*$, it is without loss of generality to consider only bids $b_j^*(v)$ where $v$ is in the support of $F$, because any other bid will be dominated by a bid of this type.

If $v'>v$ we can proceed analogously to show that the deviation is not beneficial.
\end{proof}



\subsection{Truthful VCG Payments and Allocation Probabilities}\label{sec:bidding}

	We begin by formally defining the position-specific bidding functions $b^*_j$ and computing their derivative in the valuation. We then derive a recursive formulation of the allocation probabilities, which will be used in the proofs of \lemref{lem:part-a} and \lemref{lem:part-b}. Bid $b_j^*(v)$ equals the truthful VCG payment for position $j$ given valuation $v$ and conditioned on allocation of position $j$. This quantity is equal to the sum over the differences $\beta_s-\beta_{s+1}$ multiplied by the expected value of the $s+1$-highest valuation among all agents conditioned on $v$ being the $j$-highest valuation and assuming that valuations are drawn independently from distribution~$F$. Formulaically,
\begin{align*}
	b^*_j(v) &= \sum_{s=j}^{k} (\beta_{s} - \beta_{s+1}) \int_{0}^{v}\! \frac{(n-j)!}{(n-s-1)!(s-j)!}\left(\frac{F(u)}{F(v)}\right)^{n-s-1} \left(1-\frac{F(u)}{F(v)}\right)^{s-j} \frac{f(u)}{F(v)} u \,du
\end{align*}
Using that $(1-\frac{F(u)}{F(v)})^{s-j} = \sum_{t=0}^{s-j} (-1)^t (\frac{F(u)}{F(v)})^t$ and defining $Z_{n-s+t}(v) = (\frac{1}{F(v)})^{n-s+t} \int_{0}^{v}\! F(u)^{n-s+t}\, du$ we have that
\begin{align*}  
	b^*_j(v) &= \sum_{s=j}^{k} (\beta_{s} - \beta_{s+1}) \sum_{t=0}^{s-j} (-1)^t {s-j \choose t} \frac{(n-j)!}{(n-s-1)!(s-j)!} \frac{1}{(n-s+t)} \left(v - Z_{n-s+t}(v)\right).
\end{align*}
Using that $\frac{d}{dv} (v - Z_{n-s+t}(v)) = (n-s+t) \frac{f(v)}{F(v)} Z_{n-s+t}(v)$ we obtain
\begin{align}  \label{eq:derivative}
	\frac{d}{dv} p_j(v) &= \sum_{s=j}^{k} (\beta_{s} - \beta_{s+1}) \sum_{t=0}^{s-j} (-1)^t {s-j \choose t} \frac{(n-j)!}{(n-s-1)!(s-j)!} \frac{f(v)}{F(v)} Z_{n-s+t}(v).
\end{align}

	Denote by $P_{s,m}(x)$ the probability that an agent is assigned position $s$ against $m$ opposing agents if he reports a valuation vector $x\in\rgeqk$. Then $P_{s,m}(x)$ can be written recursively as
\begin{equation}  \label{eq:recursion}
	\begin{aligned} 
		P_{1,m}(x) &= F(x_1)^m, \quad\text{and} \\
		P_{s,m}(x) &= {m \choose  m-s+1} F(x_s)^{m-s+1} \left(1-\sum_{t=1}^{s-1}P_{t,s-1}(x)\right)
	\end{aligned}
\end{equation}
The intuition behind this formulation is that the agent is assigned position $s$ if $m-s+1$ of the opposing agents have valuations smaller than $x_s$ and the agent is not assigned one of the positions $1,..,s-1$ against the remaining $s-1$ agents. 
An important observation at this point is that $P_{s,m}(x)$ does not depend on $x_\ell$ for $\ell>s$. 


\subsection{Proof of \lemref{lem:part-a}}\label{sec:lemma1}


	To prove the lemma, we write the expected utility that agent $i$ can achieve with a report $x\in\rgeqk$ given value $v$ as a sum of the contributions $T_s(x,v)=P_{s,n-1}(x)(\beta_s v - b^*_s(x_s))$ of position $s$. We then group these contributions into those of positions $s<j$, those of positions $j$ and $j+1$, and those of positions $s>j+1$, and argue for each group that their derivative in $x_j$ vanishes at $x_j=v$.

	For the contribution $\sum_{s=1}^{j-1} T(x,v)$ of positions $s<j$ this is rather straightforward, as neither the allocation probability $P_{s,n-1}(x)$, nor the utility  $\beta_s v-b^*_s(x_s)$ subject to allocation, depends on $x_j$. Hence the derivative in $x_j$ is zero everywhere, and in particular at $x_j=v$.

	To prove the claim for $T_j(x,v) + T_{j+1}(x,v)$, we first apply the recursive formulation of the allocation probabilities to compute the derivatives in $x_j$ of $T_j(x,v)$ and $T_{j+1}(x,v)$. We then observe that the derivative of $T_j(x,v) + T_{j+1}(x,v)$ vanishes at $x_j=v$ if and only if a certain differential equation involving the bids $b^*_j(v)$ and $b^*_{j+1}(v)$ is satisfied. Finally, we use the formulas for the truthful VCG payments conditioned on allocation and their derivatives to show that this differential equation is satisfied.

\begin{lemma}\label{lem:jandjplus1}
	Fix a particular agent. Assume that all other agents bid truthfully and that the agent bids truthfully on positions $j+1,\dots,k$. Then,
\[
	\frac{d}{dx_j} \big(T_j(x,v)+T_{j+1}(x,v)\big)\Bigr|_{x_j=v} = 0 .
\]
\end{lemma}
\begin{proof}
	First consider the contribution $T_j(x,v) = P_{j,n-1}(x) (\beta_{j}v - b^*_j(x_j))$ of position $j$.  By applying \eqref{eq:recursion} to $P_{j,n-1}(x)$,
\begin{align*}
	T_j(x,v) = {n-1 \choose n-j} F(x_j)^{n-j} \left(1-\sum_{t=1}^{j-1} P_{t,j-1}(x)\right) (\beta_j v - b^*_j(x_j)) ,
\end{align*}
and thus
\begin{align*}
	\frac{d}{dx_j} T_j(x,v) = {n-1 \choose n-j} \Bigg(1-\sum_{t=1}^{j-1}
		&P_{t,j-1}(x)\Bigg) \bigg((n-j)F(x_j)^{n-j-1} f(x_j) \beta_j v-\\
		&(n-j)F(x_j)^{n-j-1} f(x_j) b^*_j(x_j)-F(x_j)^{n-j} \frac{d}{dx_j} b^*_j(x_j)
	\bigg) .
\end{align*}

	Now consider the contribution $T_{j+1}(x,v) = P_{j+1,n-1}(x)(\beta_{j+1}v-b^*_{j+1}(x_{j+1})$ of position $j+1$. By applying \eqref{eq:recursion} to $P_{j+1,n-1}(x)$,
\begin{align*}
	T_{j+1}(x,v) &= {n-1 \choose n-j-1} F(v)^{n-j-1} \left(1-\sum_{t=1}^{j} P_{t,j}(x)\right) (\beta_{j+1} v - b^*_{j+1}(v)) .
\end{align*}
By pulling $P_{j,j}(x)$ out of the sum and applying \eqref{eq:recursion} to it, we obtain
\begin{multline*}
	T_{j+1}(x,v) = {n-1 \choose n-j-1} F(v)^{n-j-1} \cdot \\ \Bigg(1-\sum_{t=1}^{j-1} P_{t,j}(x) -{j \choose 1} F(x_j) \Bigg(1-\sum_{t=1}^{j-1} P_{t,j-1}(x)\Bigg) \Bigg) (\beta_{j+1} v - b^*_{j+1}(v)) ,
\end{multline*}
and thus
\begin{multline*}
	\frac{d}{dx_j} T_{j+1}(x,v) = 
	- {n-1 \choose n-j-1} F(v)^{n-j-1} {j \choose 1} f(x_j) \left(1-
		\sum_{t=1}^{j-1} P_{t,j-1}(x)\right) (\beta_{j+1} v - b^*_{j+1}(v)).
\end{multline*}

	We conclude that the derivative in $x_j$ of the contribution $T_j(x,v) + T_{j+1}(x,v)$ from positions $j$ and $j+1$ vanishes at $x_j = v$ if and only if
\begin{multline*}
	{n-1 \choose n-j} \bigg((n-j)F(v)^{n-j-1} f(v) \beta_jv -(n-j)F(v)^{n-j-1} f(v) b^*_j(v)-F(v)^{n-j} \frac{d}{dx_j} b^*_j(x_j)\Bigr|_{x_j=v}\bigg)\\
	- {n-1 \choose n-j-1} F(v)^{n-j-1} {j \choose 1} f(v) (\beta_{j+1} v - b^*_{j+1}(v)) = 0.
\end{multline*}
Using ${n-1 \choose n-j-1} {j \choose 1} = {n-1 \choose n-j} (n-j)$ to simplify and rearranging leads to the following differential equation:
\begin{align*}
	\frac{d}{dx_j} b^*_j(x_j)\Bigr|_{x_j=v} = (n-j)\frac{f(v)}{F(v)}\bigg[(\beta_{j}v -b^*_j(v)) - (\beta_{j+1}v - b^*_{j+1}(v)) \bigg].
\end{align*}

	We first observe that the $v$ parts of $b^*_j(v)$ and $b^*_{j+1}(v)$ cancel $\beta_{j}v$ and $\beta_{j+1}v$. This is the case because for $z\in\{j,j+1\}$ the $v$ part of $b^*_z(v)$ is equal to
\begin{align*}
	\sum_{s=z}^{k} (\beta_{s} - \beta_{s+1}) \underbrace{\sum_{t=0}^{s-z} (-1)^t {s-z \choose t} \frac{(n-j)!}{(n-s-1)!(s-z)!} \frac{1}{(n-s+t)}}_{=1} v 
	=\sum_{s=z}^{k} (\beta_{s} - \beta_{s+1}) v
	= \beta_{z}v.
\end{align*}

	It remains to show that $(n-j)\frac{f(v)}{F(v)}$ times the $Z$ part of $b^*_{j+1}(v)$ minus the $Z$ part of $b^*_j(v)$ is equal to the derivative in $x_j$ of $b^*_j(x_j)$ at $x_j=v$. Formulaically, the former can be expressed as
\begin{align}  \label{eq:zparts}
	& (n-j)\frac{f(v)}{F(v)}\bigg[\sum_{s=j}^{k}(\beta_s - \beta_{s+1}) \sum_{t=0}^{s-j} (-1)^t {s-j \choose t} \frac{(n-j)!}{(n-s-1)!(s-j)!} \frac{1}{(n-s+t)} Z_{n-s+t}(v)- \\
		&\sum_{s=j+1}^{k} (\beta_s - \beta_{s+1}) \sum_{t=0}^{s-j-1} (-1)^t {s-j-1 \choose t} \frac{(n-j-1)!}{(n-s-1)!(s-j-1)!} \frac{1}{(n-s+t)} Z_{n-s+t}(v) \bigg]. \notag
\end{align}
We prove the identity by showing that for all $s$ and $t$, the terms in \eqref{eq:zparts} are identical to the corresponding terms in \eqref{eq:derivative}.

	For $s=j$, the only possible value for $t$ is $t=0$, so the term in~\eqref{eq:zparts} is
\begin{align*}
	&(n-j)\frac{f(v)}{F(v)} (\beta_s - \beta_{s+1}) (-1)^t {s-j \choose t} \frac{(n-j)!}{(n-s-1)!(s-j)!} \frac{1}{(n-s+t)} Z_{n-s+t}(v).
\end{align*}
It is easy to see that this is identical to the corresponding term in \eqref{eq:derivative}, which is
\begin{align*}
	&(\beta_s - \beta_{s+1}) (-1)^t {s-j \choose t} \frac{(n-j)!}{(n-s-1)!(s-j)!} \frac{f(v)}{F(v)} Z_{n-s+t}(v).
\end{align*}

	For $s>j$ and any $t$ in the correct range the term in~\eqref{eq:zparts} is
\begin{align*}
	(n-j)\frac{f(v)}{F(v)} \bigg[ 
		&(\beta_s - \beta_{s+1}) (-1)^t {s-j \choose t} \frac{(n-j)!}{(n-s-1)!(s-j)!} \frac{1}{(n-s+t)} Z_{n-s+t}(v)-\\
		&(\beta_s - \beta_{s+1}) (-1)^t {s-j-1 \choose t} \frac{(n-j-1)!}{(n-s-1)!(s-j-1)!} \frac{1}{(n-s+t)} Z_{n-s+t}(v) \bigg] ,
\end{align*}
which using ${s-j-1 \choose t} = {s-j \choose t} \frac{s-j-t}{s-j}$ can be rewritten as
\begin{align*}
	&(\beta_s - \beta_{s+1}) (-1)^t {s-j \choose t} \frac{(n-j)!}{(n-s-1)!(s-j)!} \frac{f(v)}{F(v)} \left[ \frac{n-j}{n-s+t} - \frac{s-j-t}{n-s+t} \right] Z_{n-s+t}(v).
\end{align*}
Since $\frac{n-j}{n-s+t}-\frac{s-j-t}{n-s+t} = 1$, we obtain the corresponding term in~\eqref{eq:derivative}, which is
\begin{align*}
&(\beta_s - \beta_{s+1}) (-1)^t {s-j \choose t} \frac{(n-j)!}{(n-s-1)!(s-j)!} \frac{f(v)}{F(v)} Z_{n-s+t}(v). \tag*{\raisebox{-1ex}{\qedhere}}
\end{align*}
\end{proof}

	Next we consider the contribution $\sum_{s=j+2}^{k}T_s(x,v)$ from positions $s > j+1$. 
\begin{lemma}\label{lem:jplus2andbelow}
	Fix a particular agent. Assume that all other agents bid truthfully and that the agent bids truthfully on positions $j+1,\dots,k$. Then,
	\[
		\frac{d}{dx_j} \biggl(\sum_{s=j+2}^{k}T_s(x,v)\biggr)\Bigr|_{x_j=v} = 0 .
	\]
\end{lemma}

	Note that for position $s>j+1$ the contribution $T_s(x,v) = P_{s,n-1}(x,v) (\beta_s v - b^*_s(x_s))$ only depends on $x_j$ through the allocation probability $P_{s,n-1}(x,v)$. It therefore suffices to show that the derivative in $x_j$ of $P_{s,n-1}(x,v)$ vanishes at $x_j=v$.  We establish this claim by means of two auxiliary lemmata, which again exploit the recursive formulation of the allocation probabilities. The proofs are given in Appendices~\ref{app:auxiliary-1} and~\ref{app:auxiliary-2}.

\begin{lemma}\label{lem:auxiliary-1}
	Fix a particular agent. Assume that all other agents bid truthfully and that the agent bids truthfully on positions $j+1,\dots,k$. Then, for all $m\geq j+1$,
	\[
		\frac{d}{dx_j} \left(P_{j,m}(x) + P_{j+1,m}(x)\right)\Bigr|_{x_j=v} \,= 0 .
	\]
\end{lemma}

\begin{lemma}\label{lem:auxiliary-2}
	Fix a particular agent. Assume that all other agents bid truthfully and that the agent bids truthfully on positions $j+1,\dots,k$. Then, for all $m$ and $\ell$ such that $m \geq\ell\geq j+2$, 
	\[
		\frac{d}{dx_j} P_{\ell,m}(x) \Bigr|_{x_j=v} = 0 .
	\]
\end{lemma}

\begin{proof}[Proof of \lemref{lem:jplus2andbelow}]
	For position $s>j+1$ we first apply \eqref{eq:recursion} to $P_{s,n-1}(x)$ to obtain
\begin{align*}
	T_s(x,v) = {n-1 \choose  n-s} F(x_s)^{n-s} \left(1-\sum_{t=1}^{s-1}P_{t,s-1}(x)\right) (\beta_s v - b^*_s(x_s))
\end{align*}
and then split $\sum_{t=1}^{s-1}P_{t,s-1}(x)$ into two parts to obtain
\begin{align*}
	T_s(x,v) = {n-1 \choose  n-s} F(x_s)^{n-s} \Bigg(1-\sum_{t=1}^{j-1}P_{t,s-1}(x)- \sum_{t=j}^{s-1}P_{t,s-1}(x)\Bigg) (\beta_s v - b^*_s(x_s)) .
\end{align*}

	The derivative is thus
\begin{align*}
\frac{d}{dx_j} T_s(x,v) = {n-1 \choose  n-s} F(x_s)^{n-s} \left(-\frac{d}{dx_j}\sum_{t=j}^{s-1}P_{t,s-1}(x)\right) (\beta_s v - b^*_s(x_s)) ,
\end{align*}
and we use \lemref{lem:auxiliary-1} and \lemref{lem:auxiliary-2} to conclude that it vanishes at $x_j=v$.
\end{proof}


\subsection{Proof of \lemref{lem:part-b}} \label{sec:lemma2}

	We now turn to \lemref{lem:part-b}, and begin by recalling the results for the one-dimensional case. In this case the expected utility for report $x$ given value $v$ is equal to
\[
	\sum_{s=1}^{k} \beta_s P_{s,n-1}(x) (v-x) + \sum_{s=1}^{k} \beta_s \int_{0}^{x}\! P_{s,n-1}(t) \,dt ,
\]
which for truthful report $x=v$ simplifies to 
\begin{equation} \label{eq:onedim}
	\sum_{s=1}^{k} \beta_s \int_{0}^{v}\! P_{s,n-1}(t)\, dt .
\end{equation}
We will use this formula below to express the expected utility from positions $j+1,\dots,k$ for which both agent $i$ and the other agents report their valuations truthfully.

	To compute the derivative in $x_j$ of the expected utility we first observe that the contribution $T_s(x,v)$ is independent from $x_j$ for $s<j$, and thus 
\begin{align*}
	\frac{d}{dx_j} \bigg(\sum_{s=1}^{k} T_s(x,v)\bigg) = \frac{d}{dx_j} \bigg(\sum_{s=j}^{k} T_s(x,v)\bigg) = \frac{d}{dx_j} \bigg(T_j(x,v) + \sum_{s=j+1}^{k} T_s(x,v)\bigg).
\end{align*}
For the contribution $T_j(x,v)$ from position $j$,
\begin{align*}
	\frac{d}{dx_j} T_j(x,v) 
		&= \frac{d}{dx_j} \bigg(P_{j,n-1}(x)(\beta_j v - b^*_j(x_j))\bigg)\\
		&= \beta_j v \frac{d}{dx_j}P_{j,n-1}(x) - b^*_j(x_j) \frac{d}{dx_j}P_{j,n-1}(x) - P_{j,n-1}(x) \frac{d}{dx_j} b^*_j(x_j).
\end{align*}
For the contributions $T_s(x,v)$ from positions $s>j$ we use~\eqref{eq:onedim} to obtain
\begin{align*}
	\frac{d}{dx_j} \sum_{s=j+1}^{k} T_{s}(x,v) 
		&= \frac{d}{dx_j} \left(\sum_{s=j+1}^{k} \beta_s \int_{0}^{v} P_{s,n-1}(x_1,\dots,x_j,t,\dots,t) dt\right)\\
		&= \sum_{s=j+1}^{k} \beta_s \int_{0}^{v} \frac{d}{dx_j} P_{s,n-1}(x_1,\dots,x_j,t,\dots,t) dt.
\end{align*}
Taking the derivative in $v$ yields
\begin{align*}
	\frac{d}{dv} \left(\frac{d}{dx_j} \sum_{s=1}^{k} T_s(x,v) \right) 
		= \beta_j \frac{d}{dx_j} P_{j,n-1}(x) + \sum_{s=j+1}^{k} \beta_s \frac{d}{dx_j} P_{s,n-1}(x)
		= \frac{d}{dx_j} \sum_{s=j}^{k} \beta_s P_{s,n-1}(x).
\end{align*}

	The final step now is to argue that this expression is non-negative. That the beta fraction won increases in the report $x_j$ on position $j$ holding everything else fixed follows by an ex-post argument. If the agent was allocated a position $s<j$ then changing his reported valuation $x_j$ for position $j$ has no effect and he will still be allocated position $s$. If the agent was allocated position $s=j$, he will still be allocated this position for higher $x_j$. If the agent was allocated a position $s>j$ or no position at all, then by increasing the reported valuation $x_j$ for position $j$ he will either be allocated the same position as before or position $j$, which means that the beta fraction won will increase weakly.


\section{Conclusion and Future Work}

	In this paper we analyzed position auctions through the lens of robustness. We asked whether there exists a single mechanism that works well under complete {\em and} incomplete information settings. Specifically, we were looking to identify a mechanism that achieves the truthful VCG revenue in every efficient equilibrium.  By recalling results from prior work we were able to exclude both simplified and expressive variants of the VCG and the GSP mechanism as well as simplified variants of the GFP mechanism. We then showed that an expressive GFP mechanism indeed achieves the desired property.

	Our work has a clear message: If the goal is robustness against uncertainty about the information agents have about one another, then expressiveness beyond the type space is both necessary and sufficient. It also provides a nice counterpoint to recent work on position auctions which has highlighted the benefits of simplicity.

	An interesting question for future work is whether the message that expressiveness beyond type space is required for robustness extends to other problems. This is particularly true for the combinatorial auction problem, where simplified designs have recently received a lot of attention~\citep{CKS08a,BhRo11a,FFGL13}. 


\bibliographystyle{abbrvnat}
\bibliography{abb,gfp}


\appendix


\section{Proof of \proref{pro:complete-information}}\label{app:complete-information}

\paragraph{Proof of Part 1}

	Assume that the agents are ordered by decreasing value, \ie that $v_1 \ge v_2 \ge \dots \ge v_n$. Then in an efficient assignment agent $i$ is assigned position $i$, for $1 \le i \le k$. Denote by $u_i$ the truthful VCG utility for agent $1 \le i \le n$ and denote by $p_i$ the truthful VCG payment for position~$1 \le i \le k$. Then $u_i = \beta_i v_i - p_i$ for $1 \le i \le k$ and $u_i = 0$ for $i > k.$ 
We claim that the bid profile $b\in(\rgeqk)^n$ with
\[
	b_{i,j} = \max(\beta_j v_i - u_i, 0)
\]
for $i=1,\dots,n$ and $j=1,\dots,k$ is an equilibrium of \gfp that is efficient and yields the truthful VCG payments.

	With this bid profile, an efficient allocation assigns position $i$ to agent $i$ at price $p_i$. For the greedy allocation rule, this outcome can be obtained by letting agent $i$ point to position $i$ and breaking ties in favor of the agent that points to a given position. To see that $b$ is an equilibrium first observe that agent $i$ cannot lower his bid for position $i$ without being assigned a position other than $i$. For contradiction it thus suffices to assume that agent~$i$ has a beneficial deviation to a position $j \neq i$, such that
\[
	\beta_i v_i - p_i < \beta_j v_i - p_j - \epsilon,
\]
for every $\epsilon>0$. Here we use that agent $i$ can bid $p_j+\epsilon$ on positions $j$ and above to win one of these positions, and that he values each of them at least as highly as position $j$. The left-hand side of this inequality equals the utility of agent $i$ in the truthful equilibrium of the VCG auction, whereas the right-hand side equals the utility agent $i$ would get if he was instead assigned position $j$ at price $p_j+\epsilon$. The inequality contradicts the fact that the truthful VCG equilibrium is envy-free. 

\paragraph{Proof of Part 2}

	Consider a Nash equilibrium $b$ and assume for contradiction that it leads to an inefficient assignment. Then there exist agents $i,j$ with $v_i>v_j$ such that agent $i$ is assigned position $s$ and agent $j$ is assigned position $t<s$.

	First assume that agent $i$ bids $b_{j,t}+\epsilon$ on positions $t$ and above, which means that he wins one of these positions. Since $b$ is an equilibrium this deviation is not beneficial, \ie for every $\epsilon>0$,
\begin{equation}  \label{eq:up}
	\beta_s v_i - b_{i,s} \ge \beta_t v_i - b_{j,t} - \epsilon .
\end{equation}

	Now consider the situation where agent $j$ bids according to bid vector $b'_j$ with
\[
	b'_{j,\ell} = \begin{cases}
		b_{i,s} + \epsilon & \text{if $1\leq\ell\leq s$} \\
		0 & \text{otherwise}
	\end{cases} 
\]
for some $\epsilon>0$. We claim that with these bids agent $j$ will either be assigned a position above $s$, or will compete for position $s$ with bids that are $b_{i,s}$ or lower and will therefore be assigned position $s$. For the latter observe that agents other than $j$ who are assigned a position above $s$ when agent $j$ bids according to $b_j$ can only be assigned a higher position when agent $j$ bids according to $b'_j$. This suffices because agents other than $j$ who were assigned position $s$ or below bid at most $b_{i,s}$ on position $s$.

	Since $b$ is an equilibrium, agent $j$ does not benefit from bidding according to $b'_j$, and thus for every $\epsilon>0$, 
\begin{equation}  \label{eq:down}
	\beta_t v_j - b_{j,t} \ge \beta_s v_j - b_{i,s} - \epsilon.
\end{equation}
By adding~\eqref{eq:up} and~\eqref{eq:down} and rearranging,
\[
	\beta_s v_i + \beta_t v_j \geq \beta_s v_j + \beta_t v_i - 2 \epsilon
\]
and thus
\[
	v_j \geq v_i - \frac{2\epsilon}{\beta_t-\beta_s}
\]
for every $\epsilon>0$. This contradicts the assumption that $v_i>v_j$.

\paragraph{Proof of Part 3}

	Consider a Nash equilibrium $b=(b_1,\dots,b_n)$ and assume without loss of generality that it leads to an assignment where agent $i$ is assigned position $i$ for $i=1,\dots,k$. Further assume that the assignment is efficient, \ie that $v_1\geq v_2\geq\dots\geq v_k$. For $1\le i\le k$, agent $i+1$ does not benefit from bidding $b_{i,i}+\epsilon=p_i+\epsilon$ on position $i$ and above, so
\[
	\beta_{i+1} v_{i+1} - p_{i+1} \geq \beta_{i} v_{i+1} - p_{i} -\epsilon
\]
for every $\epsilon>0$. Thus, for every $\epsilon>0$,
\begin{align*}
	p_k &\geq \beta_k v_{k+1} - \epsilon \quad\text{and}\\
	p_i &\geq (\beta_i - \beta_{i+1}) v_{i+1} + p_{i+1} - \epsilon \quad\text{for $1\leq i<k$},
\end{align*}
which proves the claim.

\section{Proof of \lemref{lem:auxiliary-1}}\label{app:auxiliary-1}

	First consider the allocation probability $P_{j,m}(x)$ for position $j$. Applying~\eqref{eq:recursion} to $P_{j,m}(x)$ yields
\begin{align*}
	P_{j,m}(x) = {m \choose  m-j+1} F(x_j)^{m-j+1} \left(1-\sum_{t=1}^{j-1}P_{t,j-1}(x)\right) ,
\end{align*}
and thus
\begin{align*}
\frac{d}{dx_j} P_{j,m}(x) = {m \choose  m-j+1} (m-j+1) F(x_j)^{m-j+1} (m-j) f(x_j) \left(1-\sum_{t=1}^{j-1}P_{t,j-1}(x)\right) .
\end{align*}

	Now consider the allocation probability $P_{j+1,m}(x)$ of position $j+1$. Applying~\eqref{eq:recursion} to $P_{j+1,m}(x)$ yields
\begin{align*}
P_{j+1,m}(x) = {m \choose  m-j} F(v)^{m-j} \left(1-\sum_{t=1}^{j}P_{t,j}(x)\right).
\end{align*}
Pulling $P_{j,j}(x)$ out of the sum and applying~\eqref{eq:recursion} to it yields
\begin{align*}
	P_{j+1,m}(x) &= {m \choose  m-j} F(v)^{m-j} \Bigg(1-
		\sum_{t=1}^{j-1} P_{t,j}(x) - {j \choose 1} F(x_j) \left(1 - \sum_{t=1}^{j-1} P_{t,j-1}(x)\right)\Bigg) ,
\end{align*}
and thus
\begin{align*}
\frac{d}{dx_j} P_{j+1,m}(x) = {m \choose  m-j} F(v)^{m-j} \left(-{j \choose 1} f(x_j) \left(1 - \sum_{t=1}^{j-1} P_{t,j-1}(x)\right)\right) .
\end{align*}

	We conclude that the derivative in $x_j$ of $P_{j,m}(x)+P_{j+1,m}(x)$ vanishes at $x_j=v$ if and only if
\begin{align*}
{m \choose  m-j+1} (m-j+1)  F(v)^{m-j} f(v) - {m \choose m-j} F(v)^{m-j} {j \choose 1} f(v) = 0.
\end{align*}
Since ${m \choose m-j}{j \choose 1}={m \choose  m-j+1}(m-j+1)$, this is indeed the case.


\section{Proof of \lemref{lem:auxiliary-2}}\label{app:auxiliary-2}

	We prove the claim by induction over $m$, starting with $m=j+2$. In this case the only possible value of $\ell$ is $\ell=j+2$, so it suffices to show that $\frac{d}{dx_j} P_{j+2,j+2}(x)\mid_{x_j=v} = 0$. Applying \eqref{eq:recursion} to $P_{j+2,j+2}(x)$ shows that
\begin{align*}
	P_{j+2,j+2}(x) = {j+2 \choose 1} F(v) \left(1-\sum_{t=1}^{j+1}P_{t,j+1}(x)\right) .
\end{align*}
By pulling $P_{j,j+1}(x)$ and $P_{j+1,j+1}(x)$ out of the sum this can be rewritten as
\begin{align*}
	P_{j+2,j+2}(x) &= {j+2 \choose  1} F(v) \Bigg(1-\sum_{t=1}^{j-1}P_{t,j+1}(x) -\bigg(P_{j,j+1}(x)+P_{j+1,j+1}(x)\bigg)\Bigg) ,
\end{align*}
and thus
\begin{align*}
	\frac{d}{dx_j} P_{j+2,j+2}(x) = {j+2 \choose  1} F(v) \frac{d}{dx_j} \bigg(-(P_{j,j+1}(x)+P_{j+1,j+1}(x)\bigg) .
\end{align*}
Using \lemref{lem:auxiliary-1} we conclude that the derivative vanishes at $x_j=v$.

	For the inductive step assume that the claim is true for all $m'<m$. We have to show that for any $\ell$ with $m\ge l\ge j+2$ it holds that $\frac{d}{dx_j} P_{\ell,m}(x)\mid_{x_j=v} = 0$. Applying \eqref{eq:recursion} to $P_{\ell,m}(x)$ yields
\begin{align*}
	P_{\ell,m}(x) &= {m \choose  m-\ell+1} F(v)^{m-\ell+1} \left(1-\sum_{t=1}^{\ell-1}P_{t,\ell-1}(x)\right) .
\end{align*}
By splitting $\sum_{t=1}^{\ell-1}P_{t,\ell-1}(x)$ into three parts we obtain
\begin{align*}
	P_{\ell,m}(x) = {m \choose  m-\ell+1} F(v)^{m-\ell+1} \Bigg(1- \sum_{t=1}^{j-1}P_{t,\ell-1}(x)- \sum_{t=j}^{j+1}P_{t,\ell-1}(x) -\sum_{t=j+2}^{\ell-1}P_{t,\ell-1}(x)\Bigg) ,
\end{align*}
and thus
\begin{align*}
	\frac{d}{dx_j} P_{\ell,m}(x) = {m \choose  m-\ell+1} F(v)^{m-\ell+1} \Bigg(-\frac{d}{dx_j} \sum_{t=j}^{j+1}P_{t,\ell-1}(x) - \frac{d}{dx_j} \sum_{t=j+2}^{\ell-1}P_{t,\ell-1}(x)\Bigg) .
\end{align*}
Using \lemref{lem:auxiliary-1} and the induction hypothesis we conclude that the derivative again vanishes at $x_j=v$.

\end{document}